\documentclass[letterpaper, 10 pt, conference]{ieeeconf} 
\pdfoutput=1

\IEEEoverridecommandlockouts                              
\author{Bruce Lee$^1$ and Peter Seiler$^2$%
\thanks{$^1$ B. Lee is an undergraduate student at the University of
            Minnesota, Minneaolis, USA
                 \href{mailto:leex8370@umn.edu}{leex8370@umn.edu}}
\thanks{
        $^2$ P. Seiler is with Faculty of Electrical Engineering, 
            University of Michigan, Ann Arbor, USA 
                    \href{mailto:pseiler@umich.edu}{pseiler@umich.edu}
}
}

\title{\LARGE \bf
  Finite Step Performance of First-order Methods Using
  Interpolation Conditions Without Function Evaluations              
    }
\usepackage{hyperref}
\usepackage{xcolor}
\usepackage{mathtools}
\usepackage{amssymb}
\usepackage[overload]{empheq}

\usepackage{amsthm}
\newtheorem{theorem}{\textbf{Theorem}}

\newtheorem{lemma}{\textbf{Lemma}}

\newtheorem{definition}{Definition}

\newtheorem{assumption}{\textbf{Assumption}}

\newcommand{\bmtx}{\begin{bmatrix}}
\newcommand{\emtx}{\end{bmatrix}}
\newcommand{\bsmtx}{\left[ \begin{smallmatrix}} 
\newcommand{\esmtx}{\end{smallmatrix} \right]}

\usepackage{tikz}
\usetikzlibrary{shapes,matrix,arrows,calc,positioning,fit}
\tikzset{
dashedx/.style={}
}

\usepackage{pgf}
\usepackage{pgfplots}
\pgfplotsset{width=9.1cm,height=5.8cm,compat=newest}
\pgfplotscreateplotcyclelist{plain black}{%
	every mark/.append style={fill=gray,scale=0.6},mark=\star\\%
}

\usepackage{algorithm}
\usepackage{algpseudocode}
\usepackage{graphicx}
\usepackage[framemethod=tikz]{mdframed}
\definecolor{mycolor}{rgb}{1, 0, 0}

\newmdenv[innerlinewidth=0.5pt, roundcorner=4pt,linecolor=mycolor,innerleftmargin=6pt,
innerrightmargin=6pt,innertopmargin=6pt,innerbottommargin=6pt]{mybox}

\newcommand{\Tr}{\mathrm{Tr}}

\newcommand{\R}{\mathbb{R}}
\newcommand{\J}{\mathcal{J}}
\newcommand{\M}{\mathcal{M}}
\newcommand{\Hy}{\mathcal{H}}
\newcommand{\dimn}{\ell}
\def\bibtex{0}
\if\bibtex1
\else
\usepackage[style=ieee, backend=bibtex]{biblatex}
\begin{document}
\maketitle
\thispagestyle{plain}
\pagestyle{plain}
\begin{abstract}
  We present a procedure to numerically compute
  finite step worst case
  performance guarantees on a given algorithm for the unconstrained
  optimization of strongly convex functions with Lipschitz continuous
  gradients. 
  The solution method provided serves as an alternative approach to
  that derived by Taylor, Hendrickx, and Glineur in
  [Math. Prog. 161 (1-2), 2017]. The difference lies in the fact that 
  our solution uses conditions for the interpolation of a set of
  points and gradient evaluations by the gradient of a function in the
  class of interest, whereas their solution uses conditions for the
  interpolation of a set of points, gradient evaluations, and function
  evaluations by a function in the class of interest.
  The motivation for this alternative solution is that,
  in many cases, neither the algorithm nor the performance
  metric of interest rely upon function evaluations.
  The primary development is a procedure to avoid suffering from the
  factorial growth in the number of these conditions with the size of
  the set to be interpolated when solving for the worst case
  performance.
	
\end{abstract}
\section{Introduction}

In recent years, there has been efforts
to understand the worst case performance of first-order black
box optimization algorithms on specific problem classes.
The problem class that has received the most attention by far is the
unconstrained optimization of smooth, convex functions. There have
been two primary versions of worst case performance bounds defined
for this problem class. These are the finite step ($N$-step) and asymptotic
performance bounds.

We present a means of numerically solving for the worst case 
$N$-step performance of a given first-order method on strongly convex
functions with Lipschitz continuous gradients.
The unconstrained optimization is $\min_{x \in \R^d} f(x)$, where $f$ is in a specified set
of functions $\mathcal{F}$. The approach is derived by considering
conditions for a set of points $\{(y_i, u_i)\}_{i=1}^N$ to be interpolated
by the gradient 
of a function in the class of interest, i.e. $u_i = \nabla f(y_i)$ for some $f \in \mathcal{F}$.
Furthermore, we present a method to construct functions in the
problem class on which the algorithm achieves
the worst case performance. 
The contributions rely largely upon
a few key technical results, many of which are readily available in the literature
and are discussed in Section \ref{s:technical}. The primary development is then presented
in Section \ref{s:performance}. The development
may be summarized as follows:
we write conic combinations of the interpolation conditions using
doubly hyperdominant matrices, thereby reducing the dimension of the optimization problem
solved to yield the numerical performance bounds. This also enables a procedure
to reduce the number of constraints in the optimization problem solved to
generate worst case trajectories.

The solution presented
is not the first means of solving the worst case $N$-step performance
problem. The problem is formally posed in
\cite{drori2014performance}, and
upper bounds on the
worst case performance are found. In \cite{taylor2016smooth},
an exact solution to the problem is found by
developing 
necessary and sufficient
conditions for a set of points $\{(y_i, u_i, f_i)\}_{i=1}^N$ to be interpolable by
a function and its gradient, i.e $u_i = \nabla f(y_i)$ and $f_i = f(y_i)$ for
some $f \in \mathcal{F}$. As the bound is exact, it
is possible to construct a function in the class of interest
attaining the worst case performance. 

Interpolation conditions which do not
involve the function evaluations $\{f_i\}_{i=1}^N$
are presented in 
\cite{taylor2017convex}.
The drawback is that the number of conditions scales factorially
with the size of the set to be interpolated. This fact deterred the
use of such conditions for solving the worst case performance problem.
The method proposed in our paper avoids this factorial growth in the worst-case 
$N$-step performance analysis.  Specifically, we demonstrate that the
performance can be computed by solving optimization problems with $O(N^2)$
constraints. This approach also yields example functions that achieve
the computed performance.
It should be noted that the
class of functions considered in \cite{taylor2016smooth} is slightly
broader than that considered here, as it allows for the analysis of
problem classes which include functions that are not strongly convex.
Additionally, the performance measures used in
\cite{taylor2016smooth} are more general than what is considered here. In particular,
our performance measures may not include function evaluations, while theirs can. 

The solution approach outlined in this paper for finding
worst case trajectories draws inspiration from \cite{vanscoy2019trajectories},
which provides a way of constructing worst case trajectories of a linear
system constrained to satisfy a set of integral quadratic constraints. While
their problem is focused on asymptotic analysis, we present a
related result for finite horizon analysis of linear systems satisfying a set
of integral quadratic constraints. The construction is simpler than in 
the asymptotic case, and it arises almost immediately from our proof of
the lossless S-Procedure. 

The primary motivation for solving the worst case 
performance problem using interpolation conditions that do not
involve function evaluations is that many first order algorithms rely
solely upon gradient evaluations.
Therefore, introducing function evaluations as an implicit constraints
is an unnecessary step for the analysis of such algorithms. 
Furthermore, the approach proposed in our paper can be used to assess the performance of
 feedback systems with nonlinear elements, e.g. saturation.  This avoids the introduction of
 implicit constraints upon the ``function'' evaluations which have no physical meaning in these
 analyses.


Numerical comparisons between the $N$-step performance bound solved
via the approach outlined in this paper and an asymptotic performance
bound are presented in Section \ref{s:numerical}. It is interesting to note that
on the example considered, the decay
rate of the $N$-step performance bound
with $N$ almost perfectly matches the
asymptotic convergence rate. 

\vspace{0.1in}
\textbf{Notation:} The Euclidean norm of $y\in\R^d$ is denoted
$\|y\|$. The Kronecker product of two matrices $A$ and $B$ is
represented as $A \otimes B$. A symmetric, positive semidefinite
matrix $A=A^\top$ is denoted by $A\succeq 0$. Similarly, $A \succeq B$
denotes that $A - B$ is positive semidefinite. The set of symmetric
$n \times n$ matrices will be called $\mathbb{S}^n$. 

\section{Problem Statement}
\label{sec:prob}
\subsection{Terminology}



A function $f: \R^d \mapsto \R$ is \emph{convex} if the following
inequality holds for all $y_1, y_2 \in \R^d$ and
$\theta \in [0,1]$:
\begin{align}
    f(\theta y_1 + (1-\theta) y_2) \leq  \theta f(y_1)
         + (1-\theta) f(y_2).
\end{align}
Next, let $m>0$ be given and define $g:\R^d \mapsto \R$ by
$g(y):=f(y)-\frac{m}{2}\|y\|^2$. The function $f$ is
\emph{$m$-strongly convex} if $g$ is convex. Finally, $f$ has
\emph{$L$-Lipschitz gradients} for some $L<\infty$ if the function is
differentiable and the following inequality holds for all
$y_1, y_2 \in \R^d$:
\begin{align}
    \| \nabla f(y_2) - \nabla f(y_1) \| \leq L \|y_2-y_1\|.
\end{align}
This inequality implies that the gradient of $f$ is continuous.  The
class of $m$-strongly convex functions with $L$-Lipshitz gradients is
denoted by $S_{m,L}$. We will also use the notation $S_{m,L}$ when
$m=0$ and/or $L=\infty$. The case $m=0$ corresponds to functions that
are convex but not necessarily strongly convex.  The case $L=\infty$
includes functions that need not be differentiable.  In this case, the
subdifferential of $f$ at point $y$ is denoted by $\partial f(y)$.

\subsection{Performance Bounds}
Consider the unconstrained minimization of the function
$f: \R^d \mapsto \R$:
\begin{align}
    \label{eq:unconstrained}
    \min_{y \in \R^d} f(y)
\end{align}
The function $f$ is assumed to be in $S_{m,L}$ with
$0<m < L <\infty$.  This ensures that Equation~\ref{eq:unconstrained}
has a unique minimizer $y_*$.

First-order algorithms use gradient evaluations to generate a sequence
of iterates $\{y_0, y_1, \ldots\}$ that converge to the minimizer
$y_*$. We focus on algorithms that can be expressed as a linear,
time-varying (LTV) system in feedback with the gradient.  Let
$A_k\in \R^{n \times n}$, $B_k\in \R^{n \times 1}$ and
$C_k\in \R^{1 \times n}$ be given for each $k$ and define the algorithm
as:
\begin{align}
  \label{eq:LTIAlg}
  \begin{split}
    x_{k+1} &= (A_k \otimes I_d) \, x_k + (B_k \otimes I_d) \, u_k \\
    y_k & = (C_k \otimes I_d) \, x_k \\
    u_k &= \nabla f(y_k).
  \end{split}
\end{align}
Here $u_k\in \R^d$, $y_k \in \R^d$, and $x_k\in\R^{nd}$ are the input,
output, and state at iterate $k$. This includes linear, time-invariant
(LTI) algorithms as a special case, i.e. the case where $(A,B,C)$ do
  not depend on the iteration $k$.

This formulation in Equation~\ref{eq:LTIAlg} follows the work in
\cite{lessard2016analysis} and covers a large class of algorithms. For
example, $(A_k,B_k,C_k):=(1,-\alpha_k,1)$ corresponds to gradient
descent with varying stepsize:
$y_{k+1} = y_k-\alpha_k \nabla f (y_k)$. As a second example, define
the first-order algorithm with the state
$x_k:=\bmtx y_k^T & y_{k-1}^T\emtx^T$ and the following matrices:
\begin{align}
\label{eq:hb1}
  A := \bmtx (1+\beta) & -\beta \\ 1 & 0 \emtx, \,\,
  B := \bmtx -\alpha \\ 0 \emtx, \,\,
  C := \bmtx 1 & 0 \emtx.
\end{align}
This corresponds to the heavy-ball algorithm with constant parameters:
\begin{align}
\label{eq:hb2}
  y_{k+1} = y_k -\alpha \nabla f(y_k) + \beta (y_k - y_{k-1})
\end{align}
Other algorithms can be modeled as in Equation~\ref{eq:LTIAlg}
including Nesterov's accelerated method \cite{nesterov1983accelerated} and the triple
momentum method \cite{vanscoy2018TMM}.  


We assume the algorithm has an optimal state $x_*$ corresponding to
the minimizer $y_*$. Thus for each $f\in S_{m,L}$ there is a state
$x_*$ such that $x_0=x_*$ yields the iterates $x_k=x_*$ and $y_k=y_*$
for $k=0,1,\ldots$. Note that the minimizer $y_*$ satisfies
$\nabla f(y_*)=0$.  Thus the state $x_*$ must satisfy
\begin{align}
  x_* = (A_k \otimes I_d)x_*  \mbox{ and } y_* = (C_k\otimes I_d) x_*. 
\end{align}
The optimal states for gradient descent and heavy-ball are
$x_*=y_*$ and $x_* := \bmtx (y_*)^T & (y_*)^T \emtx^T$, respectively.

We consider finite-step worst-case performance
over all functions $f \in S_{m,L}$.  The performance bound of
interest is formally defined next.


\begin{definition}
  \label{def:fixedstepbnd}
  Consider a time-varying algorithm defined by $\{A_k\}_{k=0}^{N-1}$,
  $\{B_k\}_{k=0}^{N-1}$, and $\{C_k\}_{k=0}^{N}$. The
  \underline{worst-case, $N$-step performance bound} on $S_{m,L}$ is
  the smallest value of $b$ such that for any $f\in S_{m,L}$:
  \begin{align*}
    \|y_N - y_* \| \leq b \|x_0-x_*\|.
  \end{align*}
\end{definition}


This definition bounds the convergence of the iterate $y_k$ to the
minimizer $y_*$.  More general performance measures are considered in
\cite{taylor2016smooth}, including convergence of the function values
$f(y_k)$ to the minimal value $f(y_*)$.  



In the remainder of the paper we assume that $y_*=0$ and $x_*=0$.
This assumption simplifies the notation and is without loss of
generality by a coordinate shift. Specifically, assume $f\in S_{m,L}$
is minimized at $y_*\ne 0$ and the algorithm has an optimal state
$x_* \ne 0$.  Redefine the algorithm state and output to be
$\tilde{x}_k:= x_k-x_*$ and $\tilde{y}_k:=y_k-y_*$.  Define the shifted
function $\tilde{f}\in S_{m,L}$ by
$\tilde{f}( \tilde{y} ):= f( \tilde{y}+y_* )$. The shifted function
$\tilde{f}$ has $\tilde{y}=0$ as its minimizer. Moreover, the
finite-step performance bound is unchanged by this
coordinate shift.

\section{Technical Results}
\label{s:technical}
Prior to presenting the solution to the worst case performance
problem, we derive a series of technical results. Many of these
results are readily available in the literature, and references are
provided for more detailed proofs.
\subsection{Interpolation Conditions}
\label{ss:interpolation}


Consider the set 
$\{ (y_i,u_i) \}_{i\in \mathcal{I}} \subset \R^d \times \R^d$ with
$\mathcal{I}:=\{0,\ldots,R\}$.  This subsection presents conditions to
interpolate this finite set of data by the gradient of a function in
$S_{m,L}$.

\begin{definition}
  \label{def:SmLinterp}
  The set $\{ (y_i,u_i) \}_{i\in \mathcal{I}}$ is
  \underline{$S_{m,L}$ interpolable} for
  $0 \le m < L \le \infty$ if there exists 
  $f \in S_{m,L}$ such that:
  \begin{itemize}
  \item ($L<\infty$):   $\nabla f(y_i) = u_i$ for all $i \in \mathcal{I}$.
  \item ($L=\infty$):   $u_i \in \partial f(y_i)$ for all $i \in \mathcal{I}$.
  \end{itemize}
\end{definition}

Lemma~\ref{lem:SmLinterp}, stated below,  provides a necessary and sufficient condition for
$S_{m,L}$ interpolation with $0<m < L <\infty$. This result
is available as Lemma 3.25 of \cite{taylor2017convex}. The proof 
is a variation of the proof of
Theorem 4 in \cite{taylor2016smooth} which provides interpolation
conditions involving both gradient and function evaluations.
Definition~\ref{def:SmLinterp} also allows $m=0$ and/or $L=\infty$
because these cases are needed for intermediate technical results.
Cyclic monotonicity, defined next, plays a key role in the various
interpolation conditions.

\begin{definition}
  \label{def:cyclicmonotone}
  The set $\{ (y_i,u_i) \}_{i\in \mathcal{I}}$ is
  \underline{cyclically monotone} if the following inequality holds
  for any cycle of indices $\{i_0, i_1, \hdots, i_J=i_0\}$:
  \begin{align}
    \label{eq:cyclicmonotone}
    \sum_{j = 0}^{J-1} u_{i_j}^\top (y_{i_{j+1}} - y_{i_j}) \le 0
  \end{align}
\end{definition}

We first state a condition from \cite{lambert2004finite} for a finite
set of data to be $S_{0, \infty}$ interpolable.

\begin{lemma}
    \label{lem:S0infinterp}
    The set $\{ (y_i,u_i) \}_{i\in \mathcal{I}}$ is $S_{0,\infty}$
    interpolable if and only if it is cyclically monotone.
\end{lemma}
\begin{proof}

  Assume the set is $S_{0, \infty}$ interpolable, i.e. there exists
  $f\in S_{0,\infty}$ such that $u_i \in \partial f(y_i)$ for all
  $i\in \mathcal{I}$ and hence:
  \begin{align}
    \label{eq:fsubgrad}
    f(y) \ge f(y_i) + u_i^\top (y - y_i) \,\,\, \forall y \in \R^d
  \end{align}
  Apply this inequality to any cycle
  $\{i_0, i_1, \hdots, i_J=i_0\}$:
  \begin{align}
    f(y_{i_{j+1}}) \ge f(y_{i_j}) + u_{i_j}^\top (y_{i_{j+1}} - y_{i_j}) 
  \end{align}
  Sum these inequalities from $j=0$ to $j=J$ to demonstrate that
  Equation~\ref{eq:cyclicmonotone} holds. This is valid for any cycle
  and hence the finite set of data is cyclically monotone.  This
  direction of the proof is formally stated as Theorem 24.8 of
  \cite{rockafellar1970convex}.

  Conversely, assume the set is cyclically monotone. Then by Theorem
  3.4 in \cite{lambert2004finite} there is a function
  $f\in S_{0,\infty}$ that interpolates the data.
\end{proof}

If the data is cyclically monotone then there are many choices for an
interpolating function in $S_{0,\infty}$. Theorem 3.4 and Proposition
3.5 in \cite{lambert2004finite} provides an explicit construction for
an interpolating function of the form:
\begin{align}
  \label{eq:PointwiseMaxInterp}
  f(y) = \max_{i\in \mathcal{I}} \left[ \lambda_i + u_i^\top (y-y_i)\right]
\end{align}
The constants $\{\lambda_i\}_{i\in \mathcal{I}}$ are computed from a
linear program and satisfy $f(y_i) = \lambda_i$.  This construction is
a pointwise maximum of affine functions. The construction simplifies
further if the data is one-dimensional ($d=1$). Details for the case
with $d=1$ are provided in Section 8 of \cite{lambert2004finite}.


Next, two additional supporting lemmas are presented before stating the
main result.

\begin{lemma}
  \label{lem:shiftinterp}
  The set $\{ (y_i,u_i) \}_{i\in \mathcal{I}}$ is $S_{m,L}$
  interpolable if and only if the set
  $\{ (y_i,u_i-m y_i) \}_{i\in \mathcal{I}}$
  is $S_{0, L-m}$ interpolable.
\end{lemma}
\begin{proof}
  Suppose $f \in S_{m,L}$ interpolates
  $\{ (y_i,u_i) \}_{i\in \mathcal{I}}$.  Define the function $g$ by
  $g(y) := f(y) - \frac{m}{2} \|y\|^2$. Then $g$ is in $S_{0,L-m}$ and
  it interpolates $\{ (y_i,u_i-my_i) \}_{i\in
    \mathcal{I}}$. The converse follows similarly.
\end{proof}

\begin{lemma}
  \label{lem:conjinterp}
  The set $\{ (y_i,u_i) \}_{i\in \mathcal{I}}$ is $S_{0,L}$
  interpolable if and only if the set
  $\{ (u_i,y_i) \}_{i\in \mathcal{I}}$
  is $S_{1/L, \infty}$ interpolable.
\end{lemma}
\begin{proof}
  Suppose $f \in S_{m,L}$ interpolates 
  $\{ (y_i,u_i) \}_{i\in \mathcal{I}}$.  Define the conjugate of $f$
  by $f^*(u) := \sup_y u^\top y - f(y)$. It follows from Proposition
  12.60 of \cite{rockafellar1998variational} that $f^*$ interpolates
  $\{ (u_i,y_i) \}_{i\in \mathcal{I}}$ and is in $S_{1/L,\infty}$.
  Proposition 12.60 also demonstrates the converse. In particular, if
  a function in $S_{1/L, \infty}$ interpolates
  $\{ (u_i,y_i) \}_{i\in \mathcal{I}}$ then its conjugate is in
  $S_{0,L}$ and interpolates   $\{ (y_i,u_i) \}_{i\in \mathcal{I}}$. 
\end{proof}

Finally, we state the main interpolation result for $S_{m,L}$ with
$0<m < L <\infty$.

\begin{lemma} (Lemma 3.25 in \cite{taylor2017convex})
    \label{lem:SmLinterp}
    The set $\{ (y_i,u_i) \}_{i\in \mathcal{I}}$ is $S_{m,L}$
    interpolable with $0 < m < L <\infty$ if and only if 
    $\{ (L y_i-u_i, u_i-my_i) \}_{i\in \mathcal{I}}$ is cyclically
    monotone.
\end{lemma}
\begin{proof}
  The five statements below are equivalent.
  Lemma~\ref{lem:shiftinterp} implies $1\leftrightarrow 2$ and
  $3\leftrightarrow 4$. Lemma~\ref{lem:conjinterp} implies
  $2\leftrightarrow 3$ and $4\leftrightarrow 5$.
  \begin{enumerate}
  \item $\{ (y_i,u_i) \}_{i\in \mathcal{I}}$ is $S_{m,L}$ interpolable.
  \item $\{ (y_i,u_i-my_i) \}_{i\in \mathcal{I}}$ is $S_{0, L-m}$ interpolable.
  \item $\{ (u_i-my_i,y_i) \}_{i\in \mathcal{I}}$ is 
    $S_{1/(L-m), \infty}$ interpolable.
  \item  $\{ (u_i-my_i,\, \frac{1}{L-m} ( L y_i - u_i)  \}_{i\in
      \mathcal{I}}$ is $S_{0, \infty}$ interpolable.
  \item  $\{ (\frac{1}{L-m} ( L y_i - u_i), \, u_i-my_i  \}_{i\in
      \mathcal{I}}$ is $S_{0, \infty}$ interpolable.
  \end{enumerate}
  Finally, Lemma~\ref{lem:S0infinterp} implies condition 5) is
  equivalent to cyclic monotonicity of
  $\{ (L y_i-u_i, u_i-my_i) \}_{i\in \mathcal{I}}$. This step requires
  factoring the constant $\frac{1}{L-m}>0$ from each term in the
  cyclic mononotinicity constraint.
\end{proof}


If the set $\{ (L y_i-u_i, u_i-my_i) \}_{i\in \mathcal{I}}$ is
cyclically monotone then an interpolating function in $S_{m,L}$ can be
constructed as follows.  First, use Lemma~\ref{lem:S0infinterp} to
construct a function $f_4 \in S_{0,\infty}$ that interpolates the data
in Statement 4.  This can be done with a pointwise maximum of affine
functions as in Equation~\ref{eq:PointwiseMaxInterp}. Next,
interpolate the data in Statement 3 with $f_3\in S_{1/(L-m),\infty}$
defined by $f_3(y):=f_4(y)+\frac{1}{2(L-m)} \|y\|^2$. Interpolate the
data in Statement 2 by taking the conjugate: $f_2=f_3^*\in
S_{0,L-m}$. Note that evaluating $f_2(y)$ involves solving the
maximization in the definition of the conjugate. Hence the function
$f_2$ does not, in general, have an explicit expression.  Finally,
interpolate the original data with $f_1 \in S_{m,L}$ defined by
$f_1(y):=f_2(y)+\frac{m}{2} \|y\|^2$.

\subsection{S-Procedure}
\label{ss:sprocedure}


Let $M_0, M_1,\ldots, M_L \in \mathbb{S}^{\dimn}$ be given and
define quadratic functions $\sigma_i: \R^{\dimn d} \mapsto \R$ for
$i=0,\ldots,L$ by:
\begin{align}
  \label{eq:sigmai}
  \sigma_i(\eta)  := \eta^\top (M_i \otimes I_{d}) \eta
\end{align}
The matrices $M_i$ are not necessarily sign definite.  This section
reviews a technical result to answer the following question: Let
$\eta\ne 0$ be given. Does $\sigma_i(\eta) \ge 0$ for $i=1,\ldots,L$ imply
that $\sigma_0(\eta) \ge 0$?  The next lemma provides an exact linear
matrix inequality condition to answer this question.  This is known as
the (lossless) S-procedure, discussed in Section 2.6.3 of
\cite{boyd1994linear}. The formulation below is essentially from
\cite{taylor2018lyapunov} (see Theorem 7 and Appendices A/B). 

In order to ensure that the optimization problems considered throughout the remainder
of this section attain their optimal solutions, we require the following assumption:
\begin{assumption}
  \label{a:rank}
	There exists an $\eta \in \{\eta|\sigma_i(\eta)\ge 0, i=1,\hdots,L\}$ such that
	when $\eta\in \R^{\dimn d}$ is partitioned as:
  \begin{align*}
    \eta:=\bmtx \eta_1 \\ \vdots \\ \eta_\dimn \emtx \mbox{ where each } \eta_i \in \R^d
  \end{align*}
  and the blocks are stacked into a matrix
  $B:=\bmtx \eta_1, \ldots, \eta_\dimn\emtx \in \R^{d\times \dimn}$,
  $B$ has rank $\ge \dimn$.
\end{assumption}
We demonstrate in Section \ref{s:performance} that this assumption holds for the constraints
applied to solve the $N$-step performance problem.


\begin{lemma}
  \label{lem:losslessS}
  Consider the following two statements involving the quadratic
  functions in Equation~\ref{eq:sigmai}
  \begin{enumerate}
  \item If $\eta\ne 0$ and $\sigma_i(\eta) \ge 0$ for $i=1,\ldots,L$
    then $\sigma_0(\eta) \ge 0$.
  \item There exists non-negative scalars
    $\{\lambda_1, \ldots, \lambda_L\}$ such that
    $M_0 - \sum_{i=1}^L \lambda_i M_i \succeq 0$.
  \end{enumerate}
  Statement 2) implies 1). Moreover, if $d\ge \dimn$ and Assumption \ref{a:rank} holds,
  then Statement 1) implies 2).
\end{lemma}
\begin{proof}
  (2$\rightarrow$ 1) Note that 2) implies:
  \begin{align*}
    (M_0 \otimes I_{d}) - \sum_{i=1}^L \lambda_i (M_i \otimes I_{d
    }) \succeq 0 
  \end{align*}
  Multiply on the right and left by any $\eta$ and $\eta^T$ to obtain:
  \begin{align*}
    \sigma_0(\eta) \ge \sum_{i=1}^L \lambda_i \sigma_i(\eta) 
  \end{align*}
  Statement 1 follows from this inequality and using $\lambda_i \ge 0$.

  \vspace{0.1in}

  (1$\rightarrow$ 2) Assume Statement 1 holds and $d \ge \dimn$.  Consider
  the following optimization:
  \begin{align*} 
    p_* = & \min_{\eta\in \R^{\dimn d}} \sigma_0(\eta) \\
          & \mbox{subject to: } \|\eta\| = 1 \\
          & \hspace{0.62in} 
            \sigma_i(\eta) \ge 0 \mbox{ for } i = 1, \ldots, L
  \end{align*}
  Note that Statement 1 implies that $\sigma_0(\eta)\ge 0$ for any
  feasible $\eta$ for this optimization and hence $p_* \ge 0$.


  Next, partition, $\eta\in \R^{\dimn d}$ as follows:
  \begin{align*}
    \eta:=\bmtx \eta_1 \\ \vdots \\ \eta_\dimn \emtx \mbox{ where each } \eta_i \in \R^d
  \end{align*}
  Stack the partitioned blocks of $\eta$ into a matrix
  $B:=\bmtx \eta_1, \ldots, \eta_\dimn\emtx \in \R^{d\times \dimn}$. It can be
  shown, using the Kronecker product structure, that
  $\sigma_i(\eta) = \Tr(M_i B^\top B)$ for $i=0,\ldots,L$.
  As a consequence, the minimization can be equivalently written
  in terms of $G:=B^\top B \in  \mathbb{S}^\dimn$.
  \begin{align*}
    \begin{split}
    p_* = & \min_{G\in \mathbb{S}^\dimn} \Tr(M_0G) \\    
          & \mbox{subject to: } \Tr(G) = 1, \, G\succeq 0 \\
          & \hspace{0.62in} 
            \Tr(M_iG) \ge 0 \mbox{ for } i = 1, \ldots, L \\
          & \hspace{0.62in} 
            \mbox{Rank}(G) \le d
    \end{split}
  \end{align*}
  The rank constraint is satisfied due to the additional assumption
  that $d \ge \dimn$.  Hence the rank constraint can be removed to yield a
  convex, semidefinite program (SDP):
  \begin{align}
   \label{eq:primalOpt} 
   \begin{split}
   p_* = & \min_{G\in \mathbb{S}^\dimn} \Tr(M_0G) \\    
         & \mbox{subject to: } \Tr(G) = 1, \, G\succeq 0 \\
         & \hspace{0.62in} 
           \Tr(M_iG) \ge 0 \mbox{ for } i = 1, \ldots, L \\
   \end{split}
  \end{align}     
  The dual of this SDP is:
  \begin{align}
    \label{eq:dualOpt}
    \begin{split} 
    d_* = & \max_{\nu,\lambda_i \in \R} \nu \\
          & \mbox{subject to: } 
               \lambda_i \ge 0 \mbox{ for } i = 1, \ldots, L \\
          & \hspace{0.62in} 
            M_0 - \sum_{i=1}^L \lambda_i M_i \succeq \nu I_n
    \end{split}
  \end{align}
  The dual has a strictly feasible point, e.g. choose any
  $\lambda_i >0$ and $\nu$ sufficiently negative.  As a consequence,
  strong duality holds and the primal attains its optimal solution, see Section 5.9.1
  of \cite{boyd2004convex}. By Assumption \ref{a:rank}, the primal has a feasible point $G$ for
  which $G \succ 0$. Then by Proposition 6.3.2 in \cite{bertsekas1999nonlinear}, the dual problem
  attains its optimal solution.  
  As noted above,
  Statement 1 implies $p_*\ge 0$ and, by strong duality, $d_* \ge 0$.
  It follows that there exists
  $\lambda_i\ge 0$ and $\nu \ge 0$ such that
  $M_0 - \sum_{i=1}^L \lambda_i M_i \succeq \nu I_n$. Thus Statement 2
  holds.
\end{proof}



\subsection{Construction of a Worst-Case Counterexample}
\label{ss:wcexample}

Suppose Statement 2 in Lemma \ref{lem:losslessS} is false and $d \ge \dimn$. The proof of
Lemma~\ref{lem:losslessS} can be used to construct an $\eta\in
\R^{\dimn d}$
that demonstrates the falsity of Statement 1.  Specifically, if Statement 2 is false then
$M_0 - \sum_{i=1}^L \lambda_i M_i \nsucceq 0$ for all nonnegative
scalars $\{\lambda_1, \hdots, \lambda_L \}$. As a result the optimal
value of the dual problem \eqref{eq:dualOpt} satisfies $d_*<0$.
Moreover, strong duality implies $p_*=d_*<0$. Let $G_*\succeq 0$ denote
a corresponding optimal solution to the primal problem
\eqref{eq:primalOpt}. Perform a rank factorization
$G_* = B_*^\top B_*$ where $B_* \in \R^{d \times \dimn}$.  (This step may
require rows of zeros to be appended to $B_*$ to ensure it has row
dimension $d$.)  Denote the $i^{th}$ column of $B_*$ by $\eta_i$ and
define
$\eta_* := \bmtx \eta_1^\top & \eta_2^\top & \hdots & \eta_\dimn^\top
\emtx^\top$. The primal feasibility of $G_*$ implies that
$\|\eta_*\| = \Tr(G_*)=1$ and $\sigma_i(\eta_*) = \Tr(M_i G_*)\ge 0$ for
$i=1,\ldots,L$.  Moreover, $\sigma_0(\eta_*) = \Tr(M_0 G_*) = p_*<0$.
Thus $\eta_*$ is a specific vector demonstrating that Statement 1 is
false.

A key step in this construction is the numerical solution to the
primal problem \eqref{eq:primalOpt}.  This is computationally costly
if the number of constraints $L$ is large.  In some instances, a
primal optimal value $G_*$ can be obtained by first solving the dual,
and then solving the primal with only a subset of constraints.  In
particular, let $(\nu_*, \lambda_*)$ be any optimal solution to the
dual \eqref{eq:dualOpt} and define $\J := \{i \, : \,\lambda_{*,i} > 0\}$.
Define the following modified primal problem enforcing
only the subset of constraints given by $\J$:
\begin{align}
\label{eq:primalOptSmall}
  \begin{split}
    p_{*,\J}= & \min_{G\in \mathbb{S}^\dimn} \Tr(M_0G) \\    
    & \mbox{subject to: } \Tr(G) = 1, \, G\succeq 0 \\
    & \hspace{0.62in} \Tr(M_iG) \ge 0 \mbox{ for } i \in \J 
  \end{split}
\end{align}
The associated dual of this modified primal problem is:
\begin{align}
  \label{eq:dualOptSmall}
    \begin{split} 
     d_{*,\J} = & \max_{\nu,\lambda_i \in \R} \nu \\
          & \mbox{subject to: } 
                \lambda_i \ge 0 \mbox{ for } i \in \J \\
          & \hspace{0.62in}
           M_0 - \sum_{i \in \J} \lambda_i M_i \succeq \nu I_n
    \end{split}
\end{align}
This leads to the following result.

\begin{lemma}
\label{lem:fewconstraints}
If the modified primal problem \eqref{eq:primalOptSmall} has a unique
solution, $G_{*,\J}$, then this is also a solution to the original
primal problem \eqref{eq:primalOpt}.
\end{lemma}
\begin{proof}
  First note that the feasible set of the modified dual problem
  \eqref{eq:dualOptSmall} is a subset of the feasible set for the
  original dual problem \eqref{eq:dualOpt}.\footnote{Let
    $\{\lambda_i\}_{i \in \J}$ be feasible for \eqref{eq:dualOptSmall}. 
   Define $\hat{\lambda}_i:=\lambda_i$ if
    $i\in \J$, and $\hat{\lambda}_i=0$ otherwise.
    Then $\{\hat \lambda_i\}_{i=1}^{L}$ is feasible for
    the original dual problem.} Thus $d_*\ge d_{*,\J}$. Moreover, the
  optimal point $(v_*, \lambda_*)$ for \eqref{eq:dualOpt} is also
  feasible for the modified dual problem \eqref{eq:dualOptSmall}.
  This point achieves the cost $d_*$ and hence is also optimal for
  \eqref{eq:dualOptSmall}, i.e. $d_{*,\J}=d_*$. Next recall that
  strong duality holds for the original problems, i.e. $p_*=d_*$, as
  noted in the proof of Lemma~\ref{lem:losslessS}. Similarly, strong
  duality holds for the modified problems, i.e. $p_{*,\J}=d_{*,\J}$,
  because \eqref{eq:dualOptSmall} also has a strictly feasible
  point. It follows from these results that the two primal problems
  achieve the same cost $p_* = p_{*,\J}$.  Finally, the modified
  primal problem \eqref{eq:primalOptSmall} only has a subset of the
  constraints enforced for the original primal problem
  \eqref{eq:primalOpt}. Thus any optimal solution $G_*$ to the
  original primal problem \eqref{eq:primalOpt} is also optimal for the
  modified primal \eqref{eq:primalOptSmall}. In other words, the set
  of optimal points for the modified primal includes all optimal
  points for the original primal. By assumption, the modified primal
  problem has a unique optimal $G_{*,\J}$. Thus $G_{*,\J}$ is also
  optimal for the original primal problem.
\end{proof}

If the assumptions of Lemma~\ref{lem:fewconstraints} hold then a
worst-case counterexample can be constructed as follows.  First solve
the original dual problem to find the active dual variables
$\J := \{i\, : \, \lambda_{*,i} > 0\}$ for any optimal point.  Next
solve the modified primal problem \eqref{eq:primalOptSmall} to obtain
$G_{*,\J}$. If this solution is unique then $G_*=G_{*,\J}$.  The
remaining steps at the beginning of this section can be used to
construct the counterexample $\eta_*$.  The final technical result in
this section is a condition that can be used to verify if the modified
primal problem has a unique solution.  This follows from the
uniqueness and nondegeneracy results in \cite{alizadeh97}.

\begin{definition}
  Let $(\nu,\lambda)$ be any feasible point for the modified dual
  problem in \eqref{eq:dualOptSmall}. Perform the eigenvalue decomposition
  \begin{align}
    \label{eq:Zstar}
    M_0 - \sum_{i \in \J} \lambda_{i} M_i - \nu I = 
    \bmtx Q_1 & Q_2 \emtx \bmtx \Gamma & 0 \\ 0 & 0 \emtx
    \bmtx Q_1^\top \\ Q_2^\top \emtx,
  \end{align}
  where $\Gamma \in \R^{s \times s}$ is a diagonal matrix containing the
  nonzero eigenvalues. The point $(\nu,\lambda)$ is
  \underline{non-degenerate} if
  $\{I\} \cup \{Q_2^\top M_i Q_2\}_{i \in \J}$ spans
  $\mathbb{S}^{\dimn-s}$.
\end{definition}

\begin{lemma}
\label{lem:unique}
Let $(\nu_*, \lambda_*)$ be an optimal solution to
\eqref{eq:dualOptSmall} with $\lambda_{*,i}>0$ for $i\in \J$.
If this point is non-degenerate then
\eqref{eq:primalOptSmall} has a unique solution.
\end{lemma}
\begin{proof}
  Let $G_*$ and $(\nu_*,\lambda_*)$ denote optimal solutions to the
  modified primal/ dual problems (dropping the subscript $\J$ to
  simplify the notation).  These must satisfy the following
  complementary slackness conditions, as discussed in Section 5.5.2
  of \cite{boyd2004convex}:
  \begin{align}
    \label{eq:compslack1}
    & \lambda_{*,i} \Tr(M_i G_*) = 0 \mbox{ for } i \in \J \\
    \label{eq:compslack2}
    &  Z_* G_* = 0         
  \end{align}
  where $Z_* = M_0 - \sum_{i \in \J} \lambda_{*,i} M_i - \nu_* I$.
  Let $([Q_1, Q_2],\Gamma)$ be an eigendecomposition of $Z_*$ as in
  Equation~\ref{eq:Zstar}. By Lemma 1 in \cite{alizadeh97}, $G_*$ and
  $Z_*$ share eigenvectors. Thus there exists a $U_1 \in \mathbb{S}^s$ and 
  $U_2 \in \mathbb{S}^{\dimn-s}$ such that
  \begin{align*}
    G_* = \bmtx Q_1 & Q_2 \emtx \bmtx U_1 & 0 \\ 0 & U_2 \emtx 
          \bmtx Q_1^\top \\ Q_2^\top \emtx
  \end{align*}
  To achieve $Z_* G_* = 0$, we must have $\Gamma U_1 = 0$. The
  eigenvalues in $\Gamma$ are assumed to be non-zero. It follows
  that $U_1=0$ and hence
   $G_* = Q_2 U_2 Q_2^\top$.
  Primal feasibilty of $G_*$ implies $1=\Tr(G_*) = \Tr(U_2)$. In
  addition, the complementary slackness conditions in
  Equation~\ref{eq:compslack1} combined with $\lambda_{i,*}>0$ imply
  that $0=\Tr(M_i G_*)  = \Tr(Q_2^\top M_i Q_2 U_2)$ for $i\in\J$. 
  These conditions are summarized as 
  \begin{align*}
    & \Tr((Q_2^\top M_i Q_2) U_2) = 0 \mbox{ for } i \in \J \\
    & \Tr((I) U_2) = 1
  \end{align*}
  The non-degeneracy assumptions impies that these conditions
  uniquely define $U_2$ and hence $G_*$.
\end{proof}

\subsection{Conic Combinations of Cyclic Monotonicity Constraints}
\label{ss:coniccombs}
%

The technical results in the previous subsections can be used to
compute the worst case $N$-step performance bound. This will be
described in detail in the next section.  One issue is that the number
of cyclic monontonicity constraints scales with $N!$.  This subsection
provides a final technical result to alleviate this computational
growth. Specifically, it is shown that conic combinations of cyclic
monotonicity conditions may be written using doubly hyperdominant
matrices.

%

To see that this is so, first consider the set of data
$\{(y_i,u_i)\}_{i=0}^{R-1}$. By Lemma~\ref{lem:S0infinterp}, this data
is $S_{0,\infty}$ interpolable if and only if it is cyclically
monotone.  It is useful to slightly reformulate of the cyclic
monotonicity conditions.  The set $\{(y_i,u_i)\}_{i=0}^{R-1}$ is
cyclically monotone if and only if the following inequality holds for
any permutation $\{i_0, i_1, \hdots, i_{R-1}\}$ of the indices
$\{0,\ldots,R-1 \}$:
\begin{align}
  \sum_{j = 0}^{R-1} u_j^\top (y_j - y_{i_j}) \ge 0
\end{align}
Define the stacked data
$U_{0:R-1}:= \bmtx u_0^\top,\ldots,u_{R-1}^\top \emtx$ and similarly
for $Y_{0:R-1}$.  The data is cyclically monotone if and only if the
following constraints are satisfied for each $R\times R$ permutation
matrix $\{P_i\}_{i=1}^{R!}$:
\begin{align}
  \label{eq:cycmonoquad}
  U_{0:R-1}^\top  \, ((I - P_i) \otimes I_d) \, Y_{0:R-1}  \geq 0 
\end{align}
Next define the following set:
\begin{align*}
  \M_R := \left \{\sum_{i=1}^{R!} \lambda_i  (I - {P}_i) \, : \,\lambda_i
  \ge 0 \right\}.
\end{align*}
Any conic combination of the cyclic monotonicity constraints 
has the following form for some $W \in \M_R$:
\begin{align}
  U_{0:R-1}^\top \, (W \otimes I_d) \, Y_{0:R-1}  \geq 0 
\end{align}
Such conic combinations can be equivalently written with doubly
hyperdominant matrices as defined next.
\begin{definition}
  \label{def:dhd}
  A matrix is \underline{doubly hyperdominant}  if the
  off diagonal elements are nonpositive, and both the row sums and
  column sums are nonnegative. A matrix is \underline{doubly
    hyperdominant with zero excess} if it is doubly hyperdominant, and
  both the row sums and column sums are zero.
\end{definition}
Let $\Hy_R$ denote the set of $R\times R$ doubly hyperdominant
matrices.  The subset of doubly hyperdominant matrices with zero
excess is denoted $\Hy_R^0$.

\begin{lemma}
  \label{lem:dhd}
  The set $\Hy_R^0$ is equal to the set $\M_R$.
\end{lemma}
\begin{proof}
  Take any $W \in \M_R$ so that, by definition, there exists
  nonnegative $\{\lambda_i\}_{i=1}^{R!}$ such that:
  \begin{align}
    \label{eq:Mconic}
    W = \sum_{i=1}^{R!} \lambda_i (I - {P}_i) 
  \end{align}
  Each term $\lambda_i (I-P_i)$ has nonpositive off-diagonal entries
  and row/colums that sum to zero.  Thus the sum in
  Equation~\ref{eq:Mconic} is doubly hyperdominant with zero excess,
  i.e. $\M_R \subseteq \Hy_R^0$.

  Next take any $H\in \Hy_R^0$.  It follows from Theorem 3.7 in
  \cite{willems1970analysis} that $H \in \M_R$. In particular, let $r$
  be any constant greater than the diagonal elements of $H$. Then
  $H = r[I - S]$ where $S:=\frac{1}{r}(rI - H)$ has all nonnegative
  entries with row/column sums equal to $1$. $S$ is a doubly
  stochastic matrix and hence it can be decomposed as a convex
  combination of permutation matrices. This is the
  Birkhoff/von-Neumann decomposition \cite{birkhoff1946tres}. In other words, there exist permutation
  matrices $\{P_i\}_{i=1}^k$ and nonnegative $\{\alpha_i\}_{i=1}^k$
  such that $\sum_{i=1}^k \alpha_i = 1$ and
  $S = \sum_{i = 1}^k \alpha_i P_i$. The Birkhoff algorithm
  \cite{brualdi1982birkhoff} provides one specific decomposition.
  This decomposition can be performed with no more than
  $k \le R^2 - 2R +1$ terms. Define the nonnegative scalars
  $\lambda_i := r \alpha_i$ to obtain the decomposition
  $H= \sum_{i =1}^k \lambda_i (I - P_i)$. Thus $H\in \M_R$.
\end{proof}


\section{Performance Bound}
\label{s:performance}
\subsection{Formulation}

The technical results in the previous section are now used to compute
the worst case N-step performance bound.  Consider the unconstrained
minimization of  $f:\R^d \mapsto \R$ as in \eqref{eq:unconstrained}.
As noted earlier, we assume without loss of generality that $y_*=0$ is
the optimal point.  Let $\{A_k\}_{k=0}^{N-1}, \{B_k\}_{k=0}^{N-1}$,
and $\{C_k\}_{k=0}^{N}$ define a time-varying algorithm of the form
\eqref{eq:LTIAlg}.  Moreover, let
$\{x_k\}_{k=0}^{N}, \{u_k\}_{k=0}^{N-1}$, and $\{y_k\}_{k=0}^{N}$ be
the sequence of iterates generated by this algorithm starting from the
intial state $x_0$. The worst-case, $N$-step performance bound on
$S_{m,L}$ is the smallest value of $b$ such that
$\|y_N \| \leq b \|x_0\|$ holds for any $f\in S_{m,L}$.

Each iterate in the finite horizon sequences can be expressed as a
linear combination of $x_0$ and $\{u_k\}_{k=0}^{N-1}$. Define the
vector
$\eta := \bmtx x_0^\top & U_{0:N-1}^\top \emtx^\top \in \R^{n_\eta}$
with dimension $n_\eta:= (N+n)d$.  Let $R_0$ and $R_1$ be matrices
that define the following mapppings:
\begin{align}
\label{eq:R0R1}
\begin{split}
&  x_0 := (R_0 \otimes I_d) \eta, \\
&  y_N := (R_1 \otimes I_d) \eta,  
\end{split}
\end{align}
For example
$R_0:=\bmtx I_n & 0_{n \times N} \emtx \in \R^{n\times (N+n)}$.  The
matrix $R_1$ can be constructed from the state
matrices of the LTV algorithm. 

The performance bound can be expressed in terms of the matrices
defined in Equation~\ref{eq:R0R1}.  Define the quadratic function
$\sigma_0 : \R^{n_\eta} \times \R \rightarrow \R$ by
$\sigma_0(\eta,b) := \eta^\top (M_0(b) \otimes I_d) \eta$ where
$M_0(b):= b^2 R_0^\top R_0 - R_1^\top R_1$. By the definitions
in \ref{eq:R0R1}, 
$\sigma_0(\eta,b) = b^2 \|x_0\|^2 - \|y_N\|^2$. Hence
the bound $\|y_N\| \le b \|x_0\|$ is satisfied if and only
if $\sigma_0(\eta,b) \ge 0$.


Similarly, quadratic functions can be defined to encode the cyclic
monotonticity constraints.  By Lemma~\ref{lem:SmLinterp}, the data
$\{u_k\}_{k=0}^{N-1}$ and $\{y_k\}_{k=0}^{N-1}$ is $S_{m,L}$
interpolable if and only if $\{ (L y_i-u_i, u_i-my_i) \}_{i=0}^{N-1}$
is cyclically monotone. As noted earlier, we consider, without loss of
generality, the case where $y_* = 0$ is the optimal point. This occurs
when $\nabla f(0) = 0$. Thus all functions under consideration have
gradients that interpolate $(0,0)$. The cyclic monotonicity
conditions, including the point $(0,0)$, are thus given by:
\begin{align*}
 & (LY_{0:N-1}-U_{0:N-1})^\top  (Q_i \otimes I_d) (U_{0:N-1}-mY_{0:N-1}) \ge 0 \\
 & \mbox{where }  Q_i:= \bsmtx 0 \\ I \esmtx^\top  \, (I - P_i)
    \,  \bsmtx 0 \\ I \esmtx
\end{align*}
The additional blocks of zeros account for the point $(0,0)$.  This
inequality must hold for all $(N+1)\times (N+1)$ permutation matrices
$P_i$.  Let $R_{2L}$ and $R_{2m}$ be matrices that
define the following mappings:
\begin{align}  
\label{eq:R2}
\begin{split}
  & LY_{0:N-1}-U_{0:N-1}  := (R_{2L}\otimes I_d) \, \eta \\
  & U_{0:N-1}-mY_{0:N-1}  := (R_{2m}\otimes I_d)  \,  \eta 
\end{split}
\end{align}
Define the quadratic functions
$\sigma_i : \R^{n_\eta} \rightarrow \R$ by
$\sigma_i(\eta)= \eta^\top (M_i\otimes I_d) \eta$ where
\begin{align*}
  M_i:= & \frac{1}{2} \bmtx R_{2L} \\ R_{2m} \emtx^\top 
     \bmtx 0 & Q_i \\ Q_i^\top  & 0 \emtx 
   \bmtx R_{2L} \\ R_{2m} \emtx \\
\nonumber
  & \mbox{ for } i=1,\ldots,(N+1)!
\end{align*}
The cyclic monotonicity conditions are thus equivalent to
$\sigma_i(\eta)\ge 0$ for $i=1,\ldots,(N+1)!$.

A quadratic function of arbitrary dimension dimension can be constructed to verify that
 Assumption 1 holds for these conditions. The construction is based on an example due
 to Nesterov \cite{nesterov2013introductory} and details can be found in Appendix B of
 \cite{taylor2018lyapunov}.
 
\subsection{Worst Case Performance}
\label{ss:wcb}

We now state the main result which supplies a means to calculate 
the $N$-step worst case performance bound.

\begin{theorem}
  \label{thm:wcb}
  The worst case $N$-step performance of 
  the algorithm defined by $\{A_k\}_{k=0}^{N-1}$,
      $\{B_k\}_{k=0}^{N-1}$, $\{C_k\}_{k=0}^{N}$
       on $S_{m,L}$ is given by the optimal value to 
  \begin{align}
    \label{eq:wcb}
     b_*& := \min_{H \in \Hy_{N}, b \in \R} b \\
    \nonumber
     & \mbox{subject to: } 
       M_0(b) - \bsmtx R_{2L} \\ R_{2m} \esmtx^\top 
    \bsmtx 0 & H \\ H^\top & 0 \esmtx 
    \bsmtx R_{2L} \\ R_{2m} \esmtx \succeq 0.
  \end{align} 
\end{theorem}
\begin{proof}
  Recall that we may, without loss of generality, consider functions
  achieving their minimum at the origin. The worst-case $N$-step
  performance problem is to find the smallest $b$ such that
  $\|y_N\| \leq b \|x_0\|$ whenever $N$ steps of algorithm
  \eqref{eq:LTIAlg} are run on $f\in S_{m,L}$ from the intial state
  $x_0 \in \R^{nd}$.  Running \eqref{eq:LTIAlg} on functions
  $f \in S_{m,L}$ will generate sequences of iterates
  $\{(y_i,u_i)\}_{i=0}^{N-1}$ that are $S_{m,L}$ interpolable. It
  follows from Lemma~\ref{lem:SmLinterp} and the notation above that
  the iterates are $S_{m,L}$ interpolable if and only if
  $\sigma_i(\eta) \ge 0$ for $i=1,\ldots,(N+1)!$. 
  Likewise, $\|y_N\| \leq b \|x_0\|$ is equivalent to
  $\sigma_0(\eta,b)\ge 0$.  

  Thus the $N$-step worst case performance is given by the solution to
  the following optimization:
  \begin{align*}
  \begin{split} 
     & \min_{b \in \R} b \\
     & \mbox{subject to: }
     \sigma_0(\eta,b) \geq 0 \mbox{ for all } \eta \in \R^{n_\eta}
     \mbox{ satisfying} \\  
     & \hspace{0.62in} \sigma_i(\eta) \geq 0 \mbox{ for } i =
     1 \hdots (N+1)!
  \end{split}
  \end{align*} 
  By Lemma \ref{lem:losslessS}, this is equivalent to 
  \begin{align*}
  \begin{split} 
     & \min_{b \in \R, \lambda_i \in \R} b \\
          & \mbox{subject to: } \lambda_i \geq 0 \mbox{ for } i =
                  1,\hdots, (N+1)! \\
          & \hspace{0.62in} M_0(b) - \sum_{i=1}^{(N+1)!} \lambda_i M_i
          \succeq 0
  \end{split}
  \end{align*}
  The conic combinations of $M_i$ may be expressed in terms of the set
  $\M_N$ to express the above problem as
  \begin{align*}
  \begin{split} 
     & \min_{b \in \R, W \in \M_N} b \\
     & \mbox{subject to:} \\
     &   M_0(b) - \bsmtx R_{2L} \\ R_{2m} \esmtx^\top 
       \bsmtx 0 &  \bsmtx 0 \\ I \esmtx^\top W \bsmtx 0 \\ I \esmtx 
       \\ \bsmtx 0 \\ I \esmtx^\top W^\top \bsmtx 0 \\ I \esmtx & 0 \esmtx 
       \bsmtx R_{2L} \\ R_{2m} \esmtx \succeq 0.
  \end{split}
  \end{align*}  
  Apply Lemma~\ref{lem:dhd} to replace $W\in \M_N$ with
  $\hat{H}\in \Hy_N^0$. The term
  $\bsmtx 0 \\ I \esmtx^\top \hat{H} \bsmtx 0 \\ I \esmtx$ eliminates
  the first row and column of $H$. Let $H$ be the sub-matrix obtained
  by removing the first row and column of $\hat{H}$.  Then $H$ is
  doubly hyperdominant but possibly with excess, i.e.  $H\in \Hy_N$.
  This leads to the formulation in \eqref{eq:wcb}.
\end{proof}

Note that \eqref{eq:wcb} is a semidefinite program (SDP) in the
variables $b$ and $H\in \Hy_N$.  The number of independent variables
in the doubly hyperdominant matrix scales with $N^2$ even though the
cyclic monotonicity conditions involve $N!$ constraints.  This SDP 
can be efficiently solved (for moderate horizon lengths) using freely
available software.


\subsection{Worst Case Trajectory}
\label{ss:nonasymptotictraj}

The bound computed by Theorem~\ref{thm:wcb} is exact. In particular,
let $b^*$ be the optimal value found from Theorem~\ref{thm:wcb}.
Select any $b\in(0,b^*)$ and $d \geq N+n$.  There is a function
$f \in S_{m,L}$ achieving its minimum at $y_*=0$ and an initial state
$x_0$ such that the final value of the algorithm has
$\|y_N\| > b \|x_0\|$.  

The procedure in Section~\ref{ss:wcexample} can be used to construct a
feasible sequence of iterates iterpolable by such a function. In
particular, if $b<b^*$ then
$M_0 - \sum_{i=1}^L \lambda_i M_i \nsucceq 0$ for all nonnegative
scalars $\{\lambda_1, \hdots, \lambda_L \}$. As a result, the
construction at the beginning of Section~\ref{ss:wcexample}  provides
an $\eta \neq 0$ with $\sigma_i(\eta) \geq 0$ for
$i = 1, \hdots, (N+1)!$ and $\sigma_0(\eta) < 0$.  This vector can be
mapped to a sequence $\{(y_i, u_i)\}_{i=0}^{N-1}$ which is
interpolable by an $f \in S_{m,L}$ with optimal value at $y_*=0$.
Furthermore, $\sigma_0(\eta) < 0$ imples  $\|y_N\| > b    \|x_0\|$.

This procedure requires the solution to the optimization problem
\eqref{eq:primalOpt}. As noted in Section~\ref{ss:wcexample}, the
number of constraints grows factorially with $N$.  The remainder of
Section~\ref{ss:wcexample} provides a method that scales quadratically
with $N$. First let $(\nu_*, H_*)$ be optimal for the following problem:
\begin{align}
  \label{eq:violatingxdhd}
  \begin{split} 
    & \max_{\nu \in \R, H \in \Hy_{N}} \nu \\
    & \mbox{subject to: } 
    M_0(b) - \bsmtx R_{2L} \\ R_{2m} \esmtx^\top 
    \bsmtx 0 & H \\ H^\top & 0 \esmtx 
    \bsmtx R_{2L} \\ R_{2m} \esmtx \succeq  \nu I_n
  \end{split}
\end{align}
Append an additional row/and column to $H_*$ to obtain the
corresponding matrix with zero excess, $\hat{H}_*$ such that:
$H_* = \bsmtx 0 \\
I_R \esmtx^\top \hat{H}_* \bsmtx 0 \\ I_R \esmtx$. This may be done by
setting the elements of the added row/column to the negation of the
corresponding column/row sum respectively. Decompose $\hat{H}_*$ as
described in the proof of Lemma~\ref{lem:dhd}, and denote
the coefficients for the decomposition by $\lambda_*$.  Then
$(\nu_*, \lambda_*)$ solve \eqref{eq:dualOpt}.  Furthermore, if we let
$\J = \{i| \lambda_{*,i} > 0\}$ and $(\nu_*, \lambda_*)$ is
non-degenerate, then a solution to \eqref{eq:primalOptSmall} is also a
solution to \eqref{eq:primalOpt}. As $\J$ has at most
$(N+1)^2 - 2(N+1) +1$ elements, the number of constraints in
\eqref{eq:primalOpt} scales quadratically with $N$.

\section{Numerical Results}
\label{s:numerical}

Theorem \ref{thm:wcb} provides an approach to calculate the $N$-step
worst case performance $b_*$ of an algorithm on $S_{m,L}$. As noted
previously, the optimization problem \eqref{eq:wcb} is a semidefinite
program and can be efficiently solved for moderate
horizons. Section~\ref{ss:nonasymptotictraj} provides a method to
construct specific worst-case trajectories that are arbitrarily close
to $b_*$. The code to compute the worst case performance as well as to
find worst case trajectories will be available
Github\footnote{\url{https://github.com/BruceDLee/nonasymptoticOptimizationConvergence}}.
The code was tested, in part, by verifying that the performance bounds
attained match those found using the Performance Estimation Toolkit
\cite{taylor2017performance} which is based on the results in
\cite{taylor2016smooth}.

The method was used to compute the $N$-step performance bounds
for the heavy-ball (Equations~\ref{eq:hb1} and \ref{eq:hb2}).
The algorithm parameters were selected to optimize performance
on quadratic functions:
\begin{align*}
    \alpha := \frac{4}{(\sqrt{L}+ \sqrt{m})^2} 
\mbox{ and }
    \beta :=  \left(\frac{\sqrt{L}-\sqrt{m}}{ \sqrt{L}+\sqrt{m}} \right)^2 
\end{align*}
Figure~\ref{fig:hb} shows the $N$-step performance bounds for several
values of $N$ (blue-x). The left plot is for $m=1, L=10$ and the right
plot is for $m=1, L=30$.  This algorithm is known to converge
asymptotically to the optimal value if and only if the condition ratio
satisfies $ \frac{L}{m} <9+4\sqrt{5}$ \cite{badithela2019heavy}.  The
left subplot corresponds to a condition ratio below this boundary and
the finite-step bounds decay, as expected.  The right subplot
corresponds to a condition ratio above this boundary and the
finite-step bounds increase, as expected.
\begin{figure}
\begin{center}
\includegraphics[width = 0.45\textwidth]{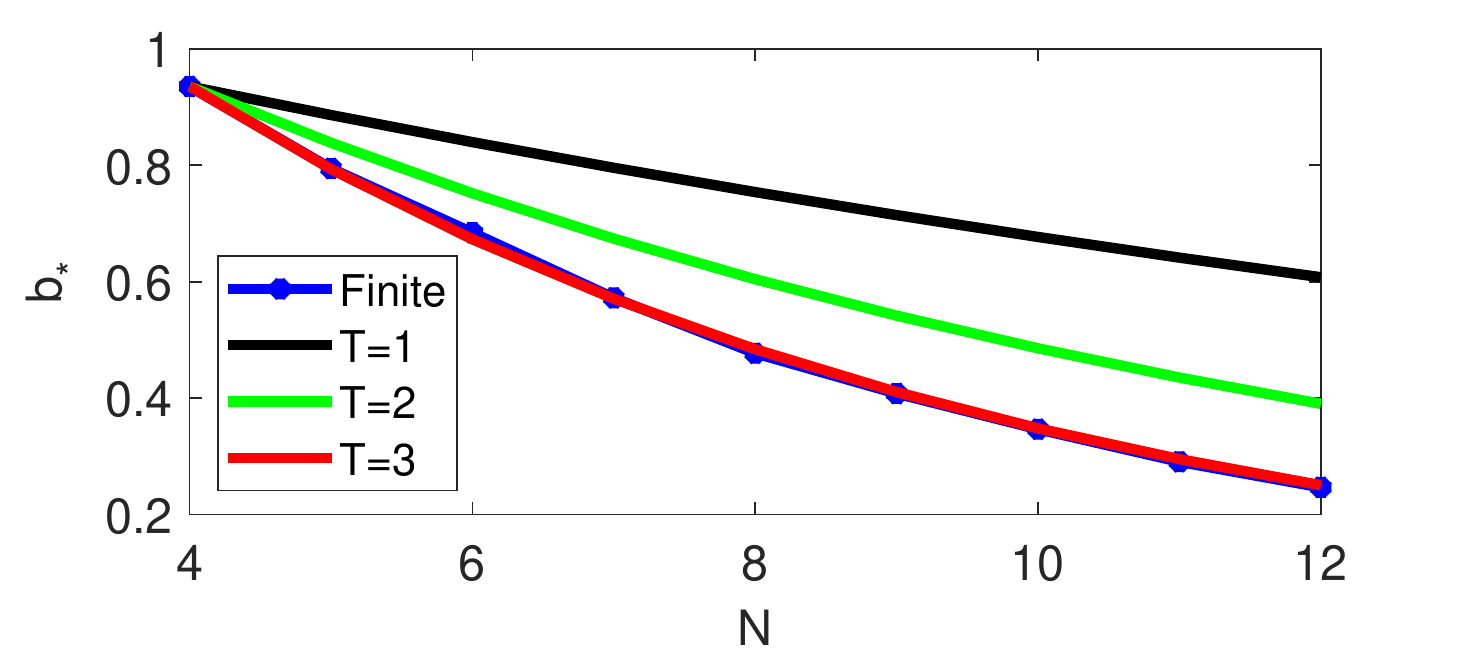}
\includegraphics[width = 0.45\textwidth]{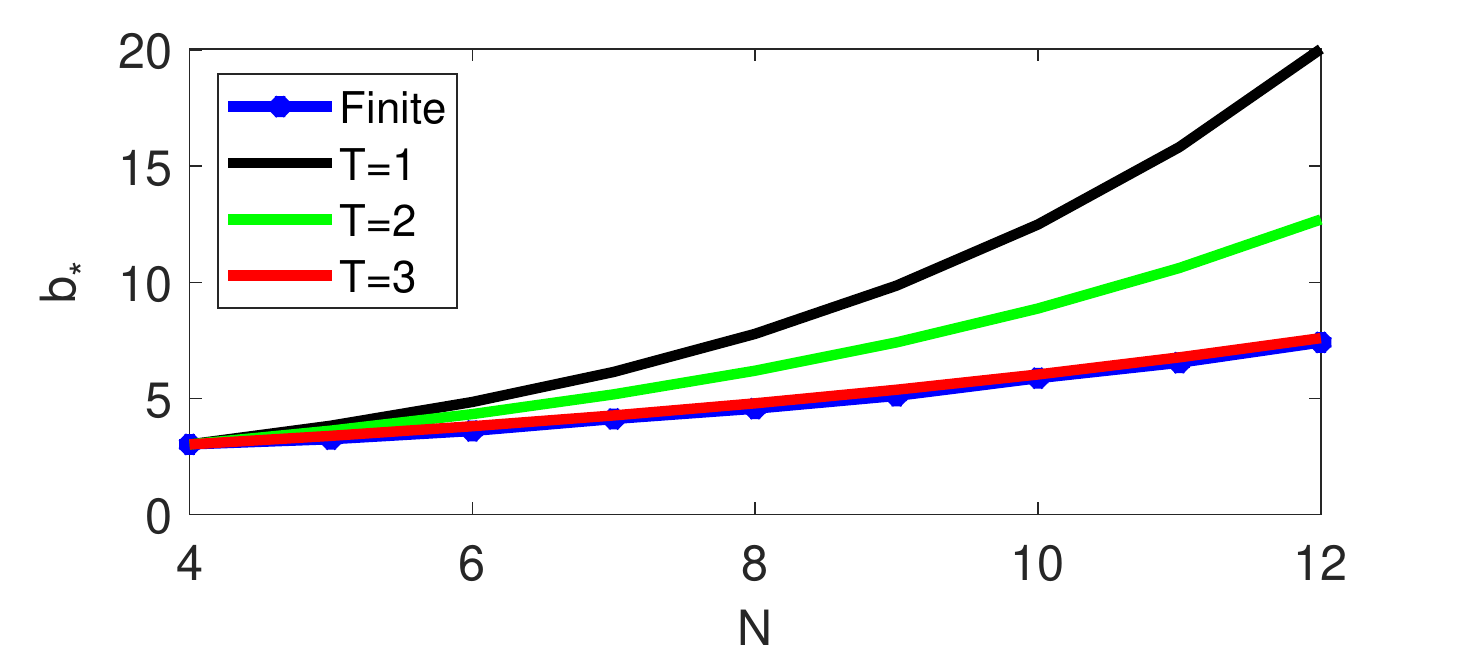}
\end{center}
\caption{Finite Step and Asymptotic $S_{m,L}$ Performance Bounds on 
        Heavy Ball Algorithm Optimized for Quadratics. Top
        plot is for $m=1$, $L=10$ and bottom plot is for
        $m=1$ and $L=30$.}
\label{fig:hb}
\end{figure}
It is interesting to compare the $N$-step performance bounds
with estimates of the asymptotic convergence rate. One asymptotic
result is briefly presented as it provides a baseline for comparison.
\begin{theorem}
  \label{thm:asymptotic}
  Consider an LTI algorithm defined by $(A,B,C)$ run on $f \in
  S_{m,L}$ achieving its minimum at the origin.
  For a positive integer $T$, define 
  \begin{align}
    \label{eq:auxseq}
    \eta_k &:= \bmtx x_{k-T} \\ U_{(k-T):(k-1)} \emtx \\ 
    z_k &:= \bmtx LY_{(k-T):(k-1)} - U_{(k-T):(k-1)} \\
    -mY_{(k-T):(k-1)} + U_{(k-T):(k-1)}\emtx
  \end{align}
  Let $(\hat{A} \otimes I_d, \hat{B} \otimes I_d, \hat{C} \otimes I_d,
                  \hat{D} \otimes I_d)$ define the state space
  system with input $u_k$, state $\eta_k$, and output $z_k$. 
  Define the optimization:
  \begin{align*}
    \rho_* & := \min_{\rho \in \R, P \in \mathbb{S}^{T+n}, H \in \Hy_{T}} \rho \\
         & \mbox{subject to:} \hspace{0.1in} P \succ 0 \\
         & 
         \hspace{0.62in} \bmtx \rho^2 P & 0 \\ 0 & 0 \emtx 
          -\bmtx \hat{A}^\top \\ \hat{B}^\top \emtx P \bmtx \hat{A}^\top \\ \hat{B}^\top \emtx^\top \\
          & \hspace{0.72in} - \bmtx \hat{C}^\top \\ \hat{D}^\top \emtx
           \bmtx 0 & H \\ H^\top & 0 \emtx 
            \bmtx \hat{C}^\top \\ \hat{D}^\top \emtx^\top \succeq 0 
  \end{align*} 
  Then there exists a constant $c>0$ such that for all 
  $k$, $\|y_k\| \leq c \rho_*^k \|\eta_0\|$. 
\end{theorem}
\begin{proof}
  Similar asymptotic performance bounds and detailed proofs are found
  in \cite{lessard2016analysis} and \cite{taylor2018lyapunov}. A
  sketch is given here.  Let $(\rho_*,P_*,H_*)$ be optimal for the
  minimization in the theorem statement.  Define
  $V(\eta):=\eta^\top P \eta$.  As in the finite-step results, the
  interpolability conditions imply:
  \begin{align*}
    \bmtx \eta_k \\ u_k \emtx^\top \bmtx \hat{C}^\top \\ \hat{D}^\top \emtx
    \bmtx 0 & H \\ H^\top & 0 \emtx 
            \bmtx \hat{C}^\top \\ \hat{D}^\top \emtx^\top \bmtx \eta_k \\ u_k \emtx \succeq 0 
  \end{align*}
  Thus the matrix inequality in the minimization implies that
  $V(\eta_{k+1}) \le \rho^2 V(\eta_k)$ for each $k$.  Iterating this
  inequality yields $\|\eta_k\| \leq c \rho_*^k \|\eta_0\|$ by setting
  $c = \sqrt{\frac{\lambda_{max}(P)}{\lambda_{min}(P)}}$. Then
  $\|y_k\| \leq \|C\| c \rho_*^k \|\eta_0\|$, where $\|C\|$ is the
  induced two norm of $\|C\|$.
\end{proof}

It should be noted that unlike the finite horizon performance bound,
the asymptotic bound is not guaranteed to be tight. In particular,
$\rho_*$ only serves as an upper bound, in general, for the asymptotic
convergence rate. As such, there are numerous variations of the
bounding approach in Theorem \ref{thm:asymptotic} which supply
different upper bounds on the asymptotic convergence, e.g.
\cite{lessard2016analysis} and \cite{taylor2018lyapunov}.

Figure~\ref{fig:hb} also shows the asymptotic rate $\rho_*$ for
several different values of $T$.  The constant $c$ for the asymptotic
curves is chosen so that each curve aligns with the finite-step bound
at $N=4$.  This allows for easier comparison. One notable aspect of
these plots is that the asymptotic rate with $T=3$ agrees, within
numerical tolerances, to the finite horizon results. The finite
horizon results are exact and hence this raises interesting
conjectures regarding the exactness of the asymptotic bounds with $T$
sufficiently large.  We also note that the results obtained using the
optimization in Theorem~\ref{thm:asymptotic} with $T=3$ are strictly
tighter than the bounds provided in \cite{lessard2016analysis} or
\cite{taylor2018lyapunov}. The examples on Github provide this
comparison and a further comparison of various conditions to
compute asymptotic rates will be explored in future work.

\section{Conclusion}
Our contribution is to provide a novel means of solving the finite
step worst case performance problem. The solution relies upon
necessary and sufficient conditions for a set of data including
points and gradient evaluations to be interpolable by the gradient
of a strongly
convex function with Lipschitz continuous gradients. Despite factorial
growth in the number of interpolation constraints with the size of the
data set, we demonstrate that the numerical solutions to the performance
bounding problem may be found from solutions to optimization problems
whose constraints grow only quadratically with the time horizon.

The motivation for solving the problem in this manner
is that a large class of algorithms do not rely upon function evaluations,
so introducing them into the constraints is unnecessary. 

It was also seen that the interpolation conditions derived 
can be extended to the case of asymptotic algorithm analysis
by straightforward application of the framework from \cite{lessard2016analysis}.
Using the asymptotic bound found by this procedure, we illustrate a
connection between the finite step and asymptotic performance bounds in Section
\ref{s:numerical}. The numerical results suggest
that the solution to the finite step performance bounding problem
may provide insight into the problem
of bounding the asymptotic convergence rate.
Further exploration of the relationship
between the problems is left as future work.
\printbibliography
\end{document}

